\documentclass{article}

\usepackage{fullpage}
\usepackage[utf8]{inputenc}
\usepackage[T1]{fontenc}
\usepackage{microtype}
\usepackage{indentfirst}

\usepackage{amsmath,amssymb,amsfonts}
\usepackage{amsthm}
\usepackage{mathtools}
\usepackage{bm}
\usepackage{commath}
\usepackage{relsize}
\usepackage{bbm}
\usepackage{nicefrac}
\usepackage{mysymbol}
\usepackage{wrapfig}

\usepackage{graphicx}
\usepackage{subfigure}
\usepackage{diagbox}
\usepackage{multirow}
\usepackage{booktabs}
\usepackage{hhline}
\usepackage{array}
\usepackage{pgfplots}
\usepackage{tikz}
\usetikzlibrary{positioning}

\usepackage{color}
\usepackage{tcolorbox}

\usepackage{algorithm}
\usepackage{algpseudocode}
\usepackage{listings}

\usepackage{cite}
\usepackage{url}

\usepackage{pifont}
\usepackage{newfloat}
\usepackage{enumerate}
\usepackage{enumitem}
\usepackage[title]{appendix}
\usepackage[font={small}]{caption}

\newtheorem{theorem}{Theorem}[section]
\newtheorem{proposition}[theorem]{Proposition}

\theoremstyle{definition}

\theoremstyle{remark}

\allowdisplaybreaks
\newcolumntype{?}{!{\vrule width 1pt}}


\title{\textbf{On the Fragility of AI-Based Channel Decoders under Small Channel Perturbations}}

\date{}

\newcommand*\samethanks[1][\value{footnote}]{\footnotemark[#1]}

\author{
Haoyu Lei\thanks{Equal contribution} \thanks{Department of Computer Science and Engineering, The Chinese University of Hong Kong, hylei22@cse.cuhk.edu.hk},
Mohammad~Jalali\samethanks[1] \thanks{Department of Computer Science and Engineering, The Chinese University of Hong Kong},
Chin Wa Lau\thanks{Huawei Technologies Co., Ltd., Theory Lab}, 
Farzan~Farnia\thanks{Department of Computer Science and Engineering, The Chinese University of Hong Kong, farnia@cse.cuhk.edu.hk}
}

\begin{document}
\maketitle

\begin{abstract} 
Recent advances in deep learning have led to AI-based error correction decoders that report empirical performance improvements over traditional belief-propagation (BP) decoding on AWGN channels. While such gains are promising, a fundamental question remains: where do these improvements come from, and what cost is paid to achieve them? In this work, we study this question through the lens of robustness to distributional shifts at the channel output. We evaluate both input-dependent adversarial perturbations (FGM and projected gradient methods under $\ell_2$ constraints) and universal adversarial perturbations that apply a single norm-bounded shift to all received vectors. Our results show that recent AI decoders, including ECCT and CrossMPT, could suffer significant performance degradation under such perturbations, despite superior nominal performance under i.i.d. AWGN. Moreover, adversarial perturbations transfer relatively strongly between AI decoders but weakly to BP-based decoders, and universal perturbations are substantially more harmful than random perturbations of equal norm. These numerical findings suggest a potential robustness cost and higher sensitivity to channel distribution underlying recent AI decoding gains.
\end{abstract}

\section{Introduction}

Reliable decoding is a cornerstone of digital communication systems. Shannon’s channel coding theorem \cite{shannon1948} establishes that rates up to channel capacity are achievable in principle, but approaching these limits in practice critically depends on the design of efficient and accurate decoding algorithms. Over the past decades, this challenge has been addressed through the development of powerful code families and iterative decoding methods, including turbo decoding \cite{berrou1993} and belief-propagation (BP) decoding for LDPC codes \cite{gallager1962,tanner1981,richardson2008}. In modern wireless systems, decoders must operate under stringent constraints on complexity, latency, energy consumption, numerical precision, and robustness to channel mismatch, making decoder design a central bottleneck in translating information-theoretic limits into practical performance \cite{3gpp38212,kaykacegilmez2020}.

Motivated by the success of deep learning in other domains, recent years have seen growing interest in applying neural network methods to channel decoding. Early learning-based approaches primarily focused on enhancing traditional decoders, for example by parameterizing belief-propagation updates or unrolling iterative decoding algorithms, as in neural BP \cite{nachmani2016}. Other works explored model-free neural decoders that learn a direct mapping from channel outputs to codewords, particularly in short blocklength regimes \cite{gruber2017}. While these approaches demonstrated promising gains in selected settings, they largely preserved the structure of classical decoding algorithms, and thus did not fundamentally depart from traditional message-passing paradigms.

More recently, transformer-based decoder architectures have been proposed that depart more substantially from classical decoding rules. Notably, the Error Correction Code Transformer (ECCT) \cite{choukroun2022ecct} introduces a transformer-based decoder with code-structure-aware attention mechanisms and reports considerable improvements over prior neural decoders. CrossMPT \cite{park2025crossmpt} further integrates masked cross-attention between magnitude and syndrome representations and reports empirical performance gains across several code families, including explicit improvements over conventional BP decoding under the considered settings. These results suggest that modern AI decoders may extract performance gains beyond traditional message-passing algorithms, raising a fundamental question: \emph{what is the source of these gains?}

In this work, we study this question through the lens of robustness to distributional shifts at the channel output. Specifically, we ask whether the nominal performance improvements reported for recent AI decoders could come at the cost of increased sensitivity to small, norm-bounded perturbations applied to the received signal. Such perturbations can be interpreted as worst-case deviations from the target AWGN model and correspond to small shifts in the distribution of channel outputs, which are commonly modeled by Wasserstein-type robustness notions \cite{esfahani2018}. Understanding decoder behavior under these shifts is essential for assessing the reliability of AI-based decoders beyond standard channel assumptions.

To address this question, we adopt tools from the adversarial robustness literature \cite{goodfellow2015,madry2018} and evaluate both input-dependent and universal perturbations applied to the received vector. We consider fast gradient methods and multi-step projected gradient attacks constrained to an $\ell_2$ norm ball, as well as universal adversarial perturbations that apply the same perturbation to all channel outputs \cite{moosavi2017}. The attacks are implemented at the AWGN channel output and optimized using gradient-based methods with randomized smoothing \cite{cohen2019} to stabilize the optimization in the presence of non-smooth decoder mappings.

Our numerical results show that standard, input-dependent adversarial perturbations can induce significant degradation in the frame error rate of recent transformer-based AI decoders, including ECCT \cite{choukroun2022ecct} and CrossMPT \cite{park2025crossmpt}.  
We further investigate universal perturbations, which are particularly relevant in practice since they do not depend on individual received vectors and could correspond to persistent interference or hardware-induced distortions. We find that universal adversarial perturbations \cite{moosavi2017} are significantly more harmful to AI decoders than random perturbations with the same $\ell_2$ norm. This observation indicates that the degradation cannot be attributed merely to increased noise power, but rather to structured vulnerabilities in the learned decoding rules.

Finally, we study the transferability of adversarial perturbations across decoders, a phenomenon widely observed in adversarial learning \cite{goodfellow2015}. Perturbations optimized for one transformer-based AI decoder often transfer effectively to another, while exhibiting much weaker transferability to BP-based decoders. Conversely, classical decoders remain comparatively robust to perturbations crafted for AI decoders. Together, these results suggest that recent AI decoding gains may rely on brittle, distribution-specific features, and they quantify a robustness cost that would be less present in traditional message-passing decoding.

\section{Related Works}

\noindent\textbf{AI and ML decoders.}
Learning-based decoding has evolved from unrolled belief propagation \cite{nachmani2016learning} to end-to-end architectures. Examples include feedback coding in Deepcode \cite{kim2018deepcode}, Turbo Autoencoders \cite{jiang2019turboae}, and meta-learning adaptations via MIND \cite{jiang2019mind}. Recently, the transformers ECCT \cite{choukroun2022ecct}, CrossMPT \cite{park2025crossmpt} and DiffMPT~\cite{lau2025interplay} have utilized attention mechanisms to improve conventional baselines. Further models include diffusion-based iterative refinement \cite{choukroun2023ddoecc}, multi-code foundation models \cite{choukroun2024foundation}, and one-step consistency models~\cite{lei2025consistency}.

\noindent\textbf{Robustness of AI decoders.}
Research indicates that learned receivers are highly sensitive to perturbations. Physical adversarial attacks have been shown to significantly degrade end-to-end autoencoders compared to classical schemes \cite{sadeghi2019physical}. In wireless security, input-agnostic perturbations have been developed for DNN-based systems \cite{bahramali2021robust}. Furthermore, deep receivers demonstrate vulnerability to constrained attacks, such as those limited by power and PAPR \cite{chen2023air}.

\noindent\textbf{Adversarial and distributional robustness in information theory.}
Classical work investigates reliable communication under worst-case models, centrally the arbitrarily varying channel (AVC) \cite{csiszar1988avc}. Related adversarial settings include oblivious channels \cite{langberg2008oblivious} and causal adversaries \cite{zhang2022causal}. Recently, capacity has been characterized for adversaries with computationally restricted observations \cite{ruzomberka2024capacity}. Parallelly, distributional robustness via Wasserstein ambiguity sets \cite{esfahani2018dro} provides a framework to connect norm-bounded perturbations to distribution shifts.

\section{Preliminaries}

We consider a binary linear error-correcting code with codebook $\mathcal{C} \subset \{0,1\}^n$, defined by a generator matrix $G \in \mathbb{F}_2^{k \times n}$ and a parity-check matrix $H \in \mathbb{F}_2^{(n-k) \times n}$ satisfying $G H^\top = 0$ over $\mathbb{F}_2$. A message $m \in \{0,1\}^k$ is encoded as the codeword $x = mG \in \mathcal{C}$ and modulated using Binary Phase-Shift Keying (BPSK), where $0 \mapsto +1$ and $1 \mapsto -1$, yielding the transmitted codeword $x_s \in \{-1,+1\}^n$. The codeword is transmitted over an Additive White Gaussian Noise (AWGN) channel, and the received signal is
\begin{equation}
y = x_s + z,
\end{equation}
where the noise vector $z \sim \mathcal{N}(0,\sigma^2 I_n)$ is drawn from a Gaussian distribution. The goal of the decoder is to estimate the transmitted codeword $\hat{x}_s$ from observed $y$.

An important quantity used in decoding is the syndrome, which is computed from a hard-demodulated version of the received signal. Specifically, we define $y_b = \operatorname{bin}(\operatorname{sign}(y))$, where $\operatorname{sign}(y_i)=+1$ if $y_i \ge 0$ and $-1$ otherwise, and $\operatorname{bin}(\cdot)$ maps $\{-1,+1\}$ to $\{1,0\}$. The corresponding syndrome is $s(y) = H y_b^\top \in \mathbb{F}_2^{n-k}$, and an error is detected whenever $s(y) \neq 0$.

Following the preprocessing strategy of \cite{bennatan2018deep}, the input to the neural decoder is constructed by concatenating the magnitude of the received signal with the syndrome,
\begin{equation}
\phi(y) = \big[\, |y|,\; s(y) \,\big],
\end{equation}
resulting in an input vector of dimension $n + (n-k)$. This representation exposes both reliability information and parity-check violations to the neural decoder and has been shown to reduce overfitting while preserving decoding performance.

\section{Robustness of AI Decoders to Input Perturbations}
We study the worst-case sensitivity of channel decoders to small perturbations applied at the channel output. Let $f$ denote a decoder (classical or AI-based) mapping a channel observation $y\in\mathbb{R}^n$ to an estimate $\hat{x_s}\in\{-1,1\}^n$. Let $x^\star\in\{-1,1\}^n$ denote the transmitted codeword, we train the AI decoder by optimizing:
\begin{equation}
    \min_{f\in\mathcal{F}}\;\: \mathbb{E} \left[ \ell\bigl(f(y), x^\star\bigr) \right],
\end{equation}
where $\ell(\cdot,\cdot)$ denotes the loss function, specifically the Binary Cross Entropy (BCE) as defined in~\cite{choukroun2022error, park2025crossmpt}.

\paragraph{Sample-wise adversarial objective.}
Given a perturbation budget $\varepsilon>0$, we define the sample-wise $\ell_2$-bounded adversarial perturbation as
\begin{equation}
\max_{\delta\in\mathbb{R}^n}\ \ell\bigl(f(y+\delta),x^\star\bigr)
\quad \text{s.t.}\quad \|\delta\|_2\le \varepsilon.
\label{eq:adv_obj_raw}
\end{equation}
Directly optimizing~\eqref{eq:adv_obj_raw} can be challenging when $f$ includes hard operations (e.g., sign, binarization, syndrome computation), which may yield non-smooth or unstable gradients. We therefore attack a smoothed surrogate as explained below.

\paragraph{Randomized smoothing of the decoder optimization objective.}
Fix $\nu>0$ and let $V\sim\mathcal{N}\bigl(0,\nu^2 I_n\bigr)$. We define the Gaussian-smoothed loss
\begin{equation}
g(y,\delta)\triangleq \mathbb{E}_{V\sim\mathcal{N}(0,\nu^2 I_n)}\Bigl[\ell\bigl(f(y+V+\delta),x^\star\bigr)\Bigr],
\label{eq:smoothed_loss}
\end{equation}
and consider the smoothed adversarial objective
\begin{equation}
\max_{\|\delta\|_2\le \varepsilon}\ g(y,\delta).
\label{eq:adv_obj_smooth}
\end{equation}
Randomized smoothing yields a stable gradient signal and enables principled gradient-based attacks.

To understand why the randomized smoothing helps, we observe in the following proposition that
Gaussian smoothing makes the loss landscape well-behaved in a \emph{dimension-free} way. This is a key technical tool for the universal-perturbation analysis in Section~\ref{subsec:uap}, where we will both (i) justify gradient-based optimization of universal objectives and (ii) derive a closed-form universal direction (UAP-PCA) from a smoothness surrogate.

\begin{proposition}[Dimension-free smoothness via Gaussian smoothing]
\label{prop:smoothness}
Assume the loss is bounded: for all $u\in\mathbb{R}^n$,
\begin{equation}
0\le \ell\bigl(f(u),x^\star\bigr)\le C.
\label{eq:bounded_loss}
\end{equation}
Define $\tilde g(u)\triangleq \mathbb{E}_{V\sim \mathcal{N}(0,\nu^2 I_n)}\bigl[\ell\bigl(f(u+V),x^\star\bigr)\bigr]$.
Then $\tilde g$ has a Lipschitz-continuous gradient with constant
\begin{equation}
\beta \ \le\ \frac{C}{\nu^2}\,\mathbb{E}\bigl[\,|Z^2-1|\,\bigr]\ <\ \frac{C}{\nu^2},
\qquad Z\sim\mathcal{N}(0,1),
\label{eq:beta_bound}
\end{equation}
i.e., $\bigl\|\nabla \tilde g(u)-\nabla \tilde g(u')\bigr\|_2 \le \frac{C}{\nu^2}\|u-u'\|_2$ for all $u,u'$.
\end{proposition}

\noindent
\begin{proof}
We defer the proof to the Appendix~\ref{app:proofs}. 
\end{proof}

\subsection{Sample-wise adversarial attacks}
\label{subsec:sample_attacks}

We consider two standard first-order attacks on the smoothed objective~\eqref{eq:adv_obj_smooth}. Both attacks estimate $\nabla_\delta g(y,\delta)$ using Monte Carlo samples $V_1,\dots,V_M\sim\mathcal{N}\bigl(0,\nu^2 I_n\bigr)$:
\begin{equation}
\nabla_{\delta} g(y,\delta)\approx \frac{1}{M}\sum_{m=1}^M \nabla_{\delta}\,\ell\Bigl(f\bigl(y+\delta+V_m\bigr),x^\star\Bigr).
\label{eq:mc_grad}
\end{equation}

\paragraph{PGD Attack}
Projected Gradient Descent (PGD)~\cite{madry2017towards} performs iterative projected gradient ascent on~\eqref{eq:adv_obj_smooth}:
\begin{equation}
\delta^{(t+1)}=\Pi_{\varepsilon}\Bigl(\delta^{(t)}+\eta\,\nabla_{\delta} g\bigl(y,\delta^{(t)}\bigr)\Bigr),
\qquad t=0,\dots,T-1,
\label{eq:pgd}
\end{equation}
where $\Pi_{\varepsilon}$ denotes Euclidean projection onto $\{\delta:\|\delta\|_2\le \varepsilon\}$. PGD refines $\delta$ over multiple steps and typically yields stronger attacks than a single gradient step.

\paragraph{FGM Attack}
The Fast Gradient Method (FGM)~\cite{goodfellow2015} is a single-step attack that linearizes $g(y,\delta)$ at $\delta=0$ and selects the maximizer of the linear approximation under the $\ell_2$ constraint:
\begin{equation}
\delta_{\mathrm{FGM}}=\varepsilon\,\frac{\nabla_{\delta} g(y,0)}{\bigl\|\nabla_{\delta} g(y,0)\bigr\|_2}.
\label{eq:fgm}
\end{equation}
FGM is computationally efficient and serves as a baseline for first-order vulnerability.

\subsection{Universal adversarial perturbations}
\label{subsec:uap}

Universal adversarial perturbations (UAPs)~\cite{moosavi2017universal} seek a single perturbation that degrades performance across many channel outputs. Let $\mathcal{D}$ denote the distribution of $(Y,X^\star)$ induced by random messages and channel noise. Define the population and empirical universal objectives
\begin{equation}
F(\delta)\triangleq \mathbb{E}_{(Y,X^\star)\sim\mathcal{D}}\bigl[g(Y,\delta)\bigr],
\qquad
\widehat F_N(\delta)\triangleq \frac{1}{N}\sum_{i=1}^N g\bigl(Y_i,\delta\bigr),
\label{eq:uap_obj}
\end{equation}
where $(Y_i,X_i^\star)_{i=1}^N$ are i.i.d.\ samples from $\mathcal{D}$.

\paragraph{UAP-Grad}
\label{subsubsec:uap_grad}
Our first universal attack, \emph{UAP-Grad}, directly applies projected gradient ascent to the empirical objective $\widehat F_N(\delta)$:
\begin{equation}
\delta^{(t+1)}=\Pi_{\varepsilon}\Bigl(\delta^{(t)}+\eta\cdot\frac{1}{N}\sum_{i=1}^N \nabla_{\delta} g\bigl(Y_i,\delta^{(t)}\bigr)\Bigr).
\label{eq:uap_grad}
\end{equation}
UAP-Grad is flexible and powerful but can be computationally demanding due to repeated gradient aggregation across many samples.

\paragraph{UAP-PCA}\label{subsubsec:uap_pca}

Our second approach, \emph{UAP-PCA}, provides a closed-form universal direction motivated by smoothness. For each sample $i$, define the smoothed per-sample function
\[
\tilde g_i(u)\triangleq \mathbb{E}_{V\sim \mathcal{N}(0,\nu^2 I_n)}\Bigl[\ell\bigl(f(u+V),X_i^\star\bigr)\Bigr],
\]
and its gradient at the sample input
\begin{equation}
q_i \triangleq \nabla \tilde g_i(Y_i)\in\mathbb{R}^n.
\label{eq:qi_def}
\end{equation}
UAP-PCA computes the universal direction from the empirical second-moment matrix
\begin{equation}
\widehat\Sigma_q \triangleq \frac{1}{N}\sum_{i=1}^N q_i q_i^\top.
\label{eq:sigmahat}
\end{equation}

Note that Proposition~\ref{prop:smoothness} implies $\beta$-smoothness for Gaussian-smoothed losses under boundedness. For UAP-PCA, we further exploit smoothness to build a quadratic lower bound in a fixed direction, then optimize that surrogate exactly. This yields an \emph{exact PCA} characterization.

\begin{proposition}
\label{prop:pca_characterization}
Assume every $\tilde g_i$ is $\beta$-smooth. For a unit direction $\delta\in\mathbb{R}^n$ with $\|\delta\|_2=1$, smoothness implies that for all $\alpha\in\mathbb{R}$,
\begin{equation}
\tilde g_i\bigl(Y_i+\alpha \delta\bigr)\ \ge\ \tilde g_i(Y_i)+\alpha\,\delta^\top q_i-\frac{\beta}{2}\alpha^2.
\label{eq:smoothness_lower_bound}
\end{equation}
Consider the surrogate maximization problem
\begin{equation}
\max_{\delta\in\mathbb{R}^n:\: \|\delta\|_2= 1}\ \frac{1}{N}\sum_{i=1}^N\ \max_{\alpha\in\mathbb{R}}
\Bigl(\alpha\,\delta^\top q_i-\frac{\beta}{2}\alpha^2\Bigr).
\label{eq:surrogate_opt}
\end{equation}
Then any maximizer $\delta^\star$ of~\eqref{eq:surrogate_opt} is a top eigenvector of $\widehat\Sigma_q$ defined in~\eqref{eq:sigmahat}.
\end{proposition}
\noindent
\begin{proof}
We defer the proof to the Appendix~\ref{app:proofs}.
\end{proof}

\subsubsection*{Finite-sample guarantees for UAP evaluation and UAP-PCA}
We next state a concentration theorem that (i) quantifies how many samples suffice to estimate the universal objective $\widehat F_N(\delta)$ for a fixed $\delta$, and (ii) guarantees stability of the empirical top principal component $\hat u_1$ of $\widehat\Sigma_q$.
\medskip
\noindent
\emph{Assumptions.} We assume (a) bounded loss $0\le \ell(\cdot)\le C$, (b) bounded gradients $\|q_i\|_2\le L$ almost surely, and (c) an eigengap $\Delta \triangleq \lambda_1(\Sigma_q)-\lambda_2(\Sigma_q)>0$, where $\Sigma_q\triangleq \mathbb{E}\bigl[q q^\top\bigr]$ and $q$ denotes a generic copy of $q_i$.

\begin{theorem}[Concentration for UAP objective and UAP-PCA]
\label{thm:uap_concentration}
Fix any $\eta\in(0,1)$. Given assumptions (a),(b),(c), then with probability at least $1-\eta$, the following hold:
\begin{align}
\bigl|\widehat F_N(\delta)-F(\delta)\bigr|
&\le
C\sqrt{\frac{\log\bigl(4/\eta\bigr)}{2N}}
\qquad \text{for any fixed $\delta$,}
\label{eq:conc_uap_obj}\\
\bigl\|\widehat\Sigma_q-\Sigma_q\bigr\|_{\mathrm{op}}
&\le
4\sqrt{2}\,L^2\sqrt{\frac{\log\bigl(4n/\eta\bigr)}{N}},
\label{eq:conc_sigma}\\
\sin\angle\bigl(\hat u_1,u_1\bigr)
&\le
\frac{4\sqrt{2}\,L^2}{\Delta}\sqrt{\frac{\log\bigl(4n/\eta\bigr)}{N}},
\label{eq:conc_pca}
\end{align}
where $u_1$ is the top eigenvector of $\Sigma_q$, $\hat u_1$ is the top eigenvector of $\widehat\Sigma_q$ (with sign chosen to maximize $|\langle \hat u_1,u_1\rangle|$), and $\angle(\cdot,\cdot)$ denotes the principal angle between one-dimensional subspaces.
\end{theorem}

\noindent
\begin{proof}
We defer the proof to the Appendix~\ref{app:proofs}.
\end{proof}

\section{Numerical Results}
\subsection{Experimental Setup}
\noindent\textbf{Datasets.} We evaluate the adversarial attack on standard error correction codes, specifically focusing on Polar and LDPC codes. Our evaluation spans various code configurations, including different code lengths ($n$) and code rates ($R = k/n$). Furthermore, we test across a range of Signal-to-Noise Ratios (SNRs), specifically targeting $E_b/N_0$ from 4 to 6 dB. We evaluate results on $10^6$ test samples for each code type.

\noindent\textbf{Evaluation Metrics.} Following existing works~\cite{choukroun2022error,parkcrossmpt}, we assess decoding performance using the standard metric: Frame Error Rate (FER), quantifying the fraction of codewords containing at least one bit error.

\noindent\textbf{Baselines.} We evaluate the adversarial attacks against a comprehensive set of decoding algorithms, categorized into: 1) Traditional Decoders: Min-Sum (MS)~\cite{fossorier1999reduced, richardson2008} and Sum-Product (SP)~\cite{gallager1962, kschischang2001factor} with maximum number of iterations equals 10. 2) Neural Decoders: State-of-the-art deep learning-based methods, specifically ECCT~\cite{choukroun2022error} and CrossMPT~\cite{parkcrossmpt}. For the neural decoders (ECCT, CrossMPT), we follow the settings specified in their publications to ensure fair comparison. Both models are implemented with a depth of $N=6$ layers and a hidden dimension of $d=128$.

\noindent\textbf{Adversarial Attack Methods.} We evaluate the decoders against both sample-level and universal adversarial attacks. For sample-level attacks, which target individual codewords, we apply Fast Gradient Method (FGM)~\cite{goodfellow2014explaining} and Projected Gradient Descent (PGD)~\cite{madry2017towards}. For universal attacks, which seek a single perturbation across the dataset, we choose Universal Adversarial Perturbation (UAP)~\cite{moosavi2017universal} and PCA-based Universal Adversarial Direction (UAP-PCA)~\cite{choi2022universal}. For all attack methods, we constrain the perturbation energy using a relative $L_2$-norm bound: $||\delta||_2 \leq \alpha||y||_2$. Additionally, we report results under random noise with the same energy constraint to demonstrate the impact of adversarial attack.

\subsection{Performance Drop after Adversarial Attack}
In our experiments, we evaluate the robustness of two baseline AI decoders, CrossMPT and ECCT, against various attack methods, including FGM, PGD, UAP, and UAP-PCA. We compare the performance in terms of the inverse Frame Error Rate \textbf{(1/FER)} before and after applying attacks. The same energy constraint of $\alpha=0.001$ is applied across all attack methods. We conduct experiments on different Polar and LDPC codes with $E_b/N_0$ ranging from 4 to 6 dB. As shown in Table~\ref{tab:attack_comparison_vertical}, although the well-trained AI decoders initially exhibit superior decoding performance, adversarial perturbations with merely $\frac{1}{1000}$ of the signal energy cause a significant degradation. While the decoders remain relatively robust to random noise, we observe a sharp decline in performance under adversarial attacks, especially at high SNRs. For decoding LDPC(49,24) at 6dB, the \textbf{(1/FER)} decreases from $5.0 \times 10^4$ to $2.6 \times 10^1$ under the PGD attack, indicating a performance drop over 1000 times. This reveals the robustness cost inherent in AI decoders under adversarial attacks. 
\begin{table}[t]
    \centering
    \caption{Decoding Performance comparison using the \textbf{FER-Inverse (1/FER)} under different attack methods.}
    \label{tab:attack_comparison_vertical}
    \renewcommand{\arraystretch}{1.125} 
    \setlength{\tabcolsep}{3.5pt}
    \resizebox{0.95\textwidth}{!}{%
        \begin{tabular}{ccc cccc cccc}
        \toprule
        \multirow{3}{*}{\textbf{Type}} & \multirow{3}{*}{\textbf{Method}} & \multirow{3}{*}{\textbf{SNR}} & \multicolumn{4}{c}{\textbf{ECCT}} & \multicolumn{4}{c}{\textbf{CrossMPT}} \\
        \cmidrule(lr){4-7} \cmidrule(lr){8-11}
        
        & & & \multicolumn{2}{c}{\textbf{Polar}} & \multicolumn{2}{c}{\textbf{LDPC}} & \multicolumn{2}{c}{\textbf{Polar}} & \multicolumn{2}{c}{\textbf{LDPC}} \\
        \cmidrule(lr){4-5} \cmidrule(lr){6-7} \cmidrule(lr){8-9} \cmidrule(lr){10-11}
        
        & & & (64,48)$\uparrow$ & (128,64)$\uparrow$ & (49,24)$\uparrow$ & (121,60)$\uparrow$ & (64,48)$\uparrow$ & (128,64)$\uparrow$ & (49,24)$\uparrow$ & (121,60)$\uparrow$ \\
        \midrule

        \multirow{6}{*}{\rotatebox{90}{\textbf{w/o attack}}} 
        & \multirow{3}{*}{\textbf{Clean}} 
          & 4 & $2.1\times10^{1}$ & $8.5\times10^{0}$ & $3.7\times10^{1}$ & $8.6\times10^{0}$ & $3.4\times10^{1}$ & $2.5\times10^{1}$ & $8.9\times10^{1}$ & $2.1\times10^{1}$ \\
        & & 5 & $1.6\times10^{2}$ & $7.4\times10^{1}$ & $3.5\times10^{2}$ & $1.4\times10^{2}$ & $2.6\times10^{2}$ & $2.4\times10^{2}$ & $1.5\times10^{3}$ & $6.1\times10^{2}$\\
        & & 6 & $2.2\times10^{3}$ & $1.4\times10^{3}$ & $9.1\times10^{3}$ & $1.1\times10^{4}$ & $4.8\times10^{3}$ & $1.0\times10^{4}$ & $5.0\times10^{4}$ & $1.7\times10^{4}$ \\
        \cmidrule(lr){2-11}
        & \multirow{3}{*}{\textbf{Random}} 
          & 4 & $2.0\times10^{1}$ & $8.1\times10^{0}$ & $3.5\times10^{1}$ & $8.2\times10^{0}$ & $3.3\times10^{1}$ & $3.0\times10^{1}$ & $7.7\times10^{1}$ & $1.9\times10^{1}$ \\
        & & 5 & $1.3\times10^{2}$ & $5.9\times10^{1}$ & $2.4\times10^{2}$ & $9.3\times10^{1}$ & $2.2\times10^{2}$ & $1.7\times10^{2}$ & $1.2\times10^{3}$ & $2.7\times10^{2}$ \\
        & & 6 & $1.7\times10^{3}$ & $6.1\times10^{2}$ & $2.0\times10^{3}$ & $3.6\times10^{3}$ & $2.8\times10^{3}$ & $4.5\times10^{3}$ & $9.8\times10^{3}$ & $3.8\times10^{3}$ \\
        \midrule
        
        \multirow{6}{*}{\rotatebox{90}{\textbf{Universal}}} 
        & \multirow{3}{*}{\textbf{UAP}} 
          & 4 & $1.7\times10^{1}$ & $5.4\times10^{0}$ & $2.8\times10^{1}$ & $5.1\times10^{0}$ & $2.8\times10^{1}$ & $1.9\times10^{1}$ & $4.0\times10^{1}$ & $1.6\times10^{1}$\\
        & & 5 & $5.3\times10^{1}$ & $3.4\times10^{1}$ & $1.1\times10^{2}$ & $3.2\times10^{1}$ & $1.0\times10^{2}$ & $7.3\times10^{1}$ & $2.0\times10^{2}$ & $5.5\times10^{1}$\\
        & & 6 & $3.6\times10^{2}$ & $9.8\times10^{1}$ & $2.7\times10^{2}$ & $2.0\times10^{2}$ & $5.7\times10^{2}$ & $5.7\times10^{2}$ & $6.4\times10^{2}$ & $5.1\times10^{2}$\\
        \cmidrule(lr){2-11}
        & \multirow{3}{*}{\textbf{UAP-PCA}} 
          & 4 & $1.5\times10^{1}$ & $5.0\times10^{0}$ & $2.4\times10^{1}$ & $4.7\times10^{0}$ & $2.4\times10^{1}$ & $1.3\times10^{1}$ & $3.5\times10^{1}$ & $1.4\times10^{1}$ \\
        & & 5 & $4.5\times10^{1}$ & $2.0\times10^{1}$ & $9.8\times10^{1}$ & $2.1\times10^{1}$ & $6.9\times10^{1}$ & $5.6\times10^{1}$ & $1.6\times10^{2}$& $4.5\times10^{1}$\\
        & & 6 & $1.5\times10^{2}$ & $5.8\times10^{1}$ & $1.6\times10^{1}$ & $1.1\times10^{2}$ & $2.4\times10^{2}$ & $1.9\times10^{2}$ & $3.4\times10^{2}$ & $1.9\times10^{2}$\\
        \midrule

        \multirow{6}{*}{\rotatebox{90}{\textbf{Sample-level}}} 
        & \multirow{3}{*}{\textbf{FGM}} 
          & 4 & $1.4\times10^{1}$ & $4.9\times10^{0}$ & $2.0\times10^{1}$ & $4.1\times10^{0}$ & $2.2\times10^{1}$ & $8.4\times10^{0}$ & $2.9\times10^{1}$ & $1.3\times10^{1}$ \\
        & & 5 & $3.2\times10^{1}$ & $1.4\times10^{1}$ & $5.5\times10^{1}$ & $1.5\times10^{1}$ & $3.8\times10^{1}$ & $4.7\times10^{1}$ & $1.2\times10^{2}$ & $3.1\times10^{1}$\\
        & & 6 & $7.1\times10^{1}$ & $2.8\times10^{1}$ & $9.3\times10^{1}$ & $6.3\times10^{1}$ & $7.5\times10^{1}$ & $1.2\times10^{2}$& $1.7\times10^{2}$ & $1.1\times10^{2}$\\
        \cmidrule(lr){2-11}
        & \multirow{3}{*}{\textbf{PGD}} 
          & 4 & $7.8\times10^{0}$ & $3.0\times10^{0}$ & $6.0\times10^{0}$ & $3.3\times10^{0}$ & $9.0\times10^{0}$ & $3.5\times10^{0}$ & $1.0\times10^{1}$ & $6.8\times10^{0}$\\
        & & 5 & $1.8\times10^{1}$ & $5.4\times10^{0}$ & $9.3\times10^{0}$ & $8.3\times10^{0}$ & $1.7\times10^{1}$ & $6.9\times10^{0}$ & $1.5\times10^{1}$ & $1.8\times10^{1}$\\
        & & 6 & $3.3\times10^{1}$ & $8.3\times10^{0}$ & $1.5\times10^{1}$ & $1.7\times10^{1}$ & $3.2\times10^{1}$ & $1.9\times10^{1}$ & $2.6\times10^{1}$ & $3.0\times10^{1}$\\
        
        \bottomrule
        \end{tabular}
    }
\end{table}

\subsection{Ablation Study: Change of Energy Constraint}
This ablation study investigates how perturbation strength affects performance by varying the energy constraint $\alpha$ ($||\delta||_2 \leq \alpha||y||_2$) between $10^{-4}$ and $10^{-1}$. Evaluating ECCT on Polar(128,64) across different SNRs, Figure~\ref{fig:radius} illustrates that degradation becomes progressively more severe with larger energy constraints, an effect that is especially pronounced at high SNRs. Detailed Results for LDPC codes refer to Appendix~\ref{appendix:ablation}.
\begin{figure}[htbp]
    \centering
    \includegraphics[width=0.96\linewidth]{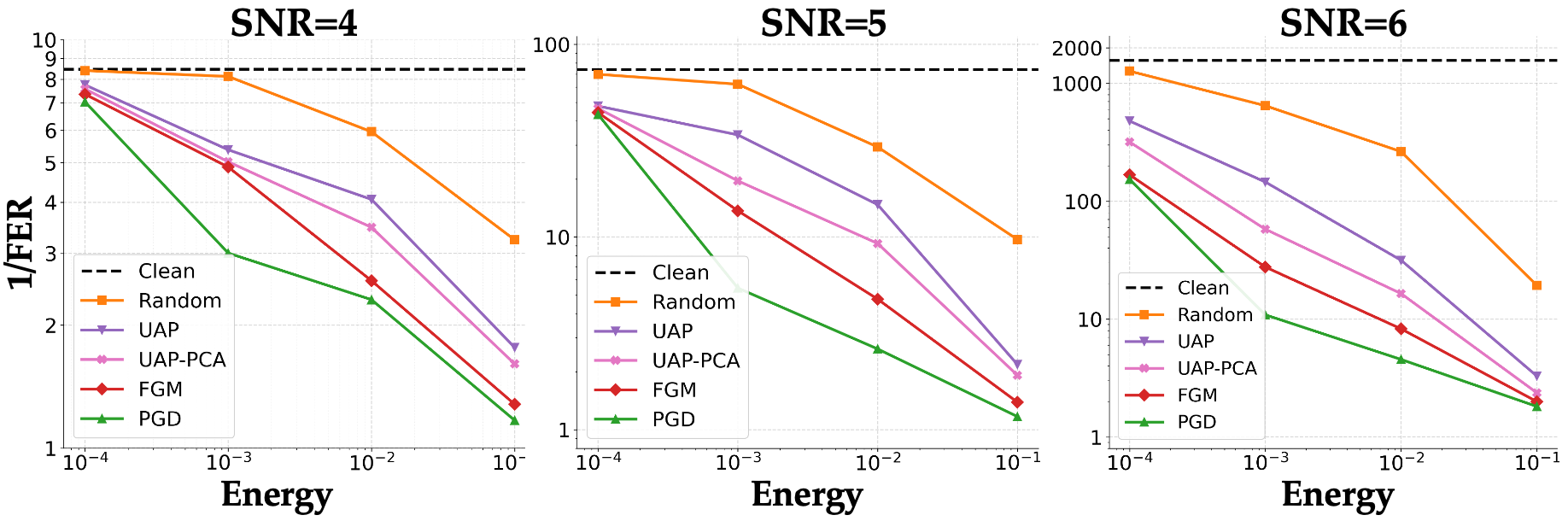}
    \caption{Comparison of \textbf{FER-Inverse (1/FER)} under various energy constraints for ECCT on Polar(128,64).}
    \label{fig:radius}
\end{figure}

\subsection{Transferability of Adversarial Perturbation}
We further evaluate the transferability of adversarial perturbations and assess the robustness of traditional decoders. We generate optimal UAP and PGD perturbations based on the ECCT model for the LDPC(121,60) code, under an energy constraint of $\alpha=0.001$. The decoding performance, measured by \textbf{1/FER}, is tested across AI decoders (ECCT, CrossMPT) and traditional decoders (Min-Sum, Sum-Product) under four conditions: clean, random noise, and the transferred UAP and PGD perturbations. As shown in Table~\ref{tab:transfer_origin}, while AI decoders achieve superior performance on clean data, they exhibit significant fragility under attack. Notably, the perturbations crafted for ECCT successfully transfer to CrossMPT, causing a drastic performance degradation (e.g., from $1.7 \times 10^4$ to $5.5 \times 10^1$ at 6 dB under PGD). In contrast, traditional decoders demonstrate robustness, maintaining consistent performance levels even with adversarial distribution shifts. Additional Results, including transferability to diffusion-based DDECC~\cite{choukroun2023ddoecc} on both LDPC and Polar codes, refer to Appendix~\ref{appendix:transfer}.
\begin{table}[H]
    \centering
    \caption{Evaluation of attack transferability measured by \textbf{FER-Inverse (1/FER)} across decoders on LDPC(121,60).}
    \renewcommand{\arraystretch}{1.125}
    \setlength{\tabcolsep}{4pt} 
    
    \resizebox{0.6\textwidth}{!}{%
        \begin{tabular}{c c l c c c}
            \toprule
            \textbf{Type} & \textbf{Decoder} & \textbf{Attack} & \textbf{4 dB} & \textbf{5 dB} & \textbf{6 dB} \\
            \midrule
            
            \multirow{8}{*}{\rotatebox{90}{\textbf{AI Decoder}}} 
            
            & \multirow{4}{*}{ECCT}     
              & Clean & $8.6 \times 10^{0}$ & $1.4 \times 10^{2}$ & $1.1 \times 10^{4}$ \\
            & & Random & $8.2 \times 10^{0}$ & $9.3 \times 10^{1}$ & $3.6 \times 10^{3}$ \\
            & & \textbf{UAP}   & $5.1 \times 10^{0}$ & $3.2 \times 10^{1}$ & $8.0 \times 10^{2}$ \\
            & & \textbf{PGD}   & $3.3 \times 10^{0}$ & $8.3 \times 10^{0}$ & $1.7 \times 10^{1}$ \\ 
            \cmidrule(l){2-6}
            
            & \multirow{4}{*}{CrossMPT} 
              & Clean & $2.1 \times 10^{1}$ & $6.1 \times 10^{2}$ & $1.7 \times 10^{4}$ \\
            & & Random & $1.9 \times 10^{1}$ & $2.7 \times 10^{2}$ & $3.8 \times 10^{3}$ \\
            & & \textbf{UAP}   & $1.7 \times 10^{1}$ & $6.7 \times 10^{1}$ & $8.8 \times 10^{2}$ \\
            & & \textbf{PGD}   & $9.2 \times 10^{0}$ & $3.2 \times 10^{1}$ & $5.5 \times 10^{1}$ \\

            \midrule

            \multirow{8}{*}{\rotatebox{90}{\textbf{Traditional Decoder}}} 
            
            & \multirow{4}{*}{Min-sum}       
              & Clean & $5.4 \times 10^{0}$ & $5.1 \times 10^{1}$ & $2.0 \times 10^{3}$ \\
            & & Random & $5.4 \times 10^{0}$ & $5.1 \times 10^{1}$ & $1.9 \times 10^{3}$ \\ 
            & & \textbf{UAP}   & $5.3 \times 10^{0}$ & $5.0 \times 10^{1}$ & $1.7 \times 10^{3}$ \\
            & & \textbf{PGD}   & $5.0 \times 10^{0}$ & $4.3 \times 10^{1}$ & $1.3 \times 10^{3}$ \\
            \cmidrule(l){2-6}
            
            & \multirow{4}{*}{Sum-product}  
              & Clean & $1.6 \times 10^{1}$ & $4.9 \times 10^{2}$ & $1.5 \times 10^{4}$ \\
            & & Random & $1.6 \times 10^{1}$ & $4.8 \times 10^{2}$ & $1.4 \times 10^{4}$ \\
            & & \textbf{UAP}   & $1.5 \times 10^{1}$ & $4.1 \times 10^{2}$ & $1.1 \times 10^{4}$ \\
            & & \textbf{PGD}   & $1.3 \times 10^{1}$ & $3.3 \times 10^{2}$ & $1.0 \times 10^{4}$ \\
            
            \bottomrule
        \end{tabular}
    }
    \label{tab:transfer_origin}
\end{table}

\section{Conclusion and Limitations}
This work indicates that although recent AI-based decoders can outperform traditional baselines under nominal AWGN conditions, their performance may be more sensitive to small, worst-case input perturbations crafted for these models. Across both input-dependent and universal attacks, gains under the standard channel model do not necessarily persist under adversarial perturbations, while BP decoders appear less affected by AI-decoder-targeted perturbations. Future work should explore robustness-aware designs and broader channel models.

\bibliographystyle{unsrt}
\bibliography{main}

\appendices
\onecolumn
\section{Proofs}
\label{app:proofs}

\subsection{Proof of Proposition~\ref{prop:smoothness}}
\label{app:proof_smoothness}

To prove Proposition~\ref{prop:smoothness}, let $V\sim\mathcal{N}\bigl(0,\nu^2 I_n\bigr)$ and define
\[
\tilde g(u)\triangleq \mathbb{E}\Bigl[\ell\bigl(f(u+V),x^\star\bigr)\Bigr].
\]
The application of Stein's lemma shows that
\begin{equation}
\nabla \tilde g(u)=\frac{1}{\nu^2}\,\mathbb{E}\Bigl[\ell\bigl(f(u+V),x^\star\bigr)\,V\Bigr].
\label{eq:stein_grad}
\end{equation}
Considering the second-order differentiation and applying the second-order Stein identity will further yield:
\begin{equation}
\nabla^2 \tilde g(u)=\frac{1}{\nu^4}\,\mathbb{E}\Bigl[\ell\bigl(f(u+V),x^\star\bigr)\,\bigl(VV^\top-\nu^2 I_n\bigr)\Bigr].
\label{eq:stein_hess}
\end{equation}
Fix any unit vector $a\in\mathbb{R}^n$ with $\|a\|_2=1$. Taking the quadratic form of~\eqref{eq:stein_hess},
\begin{equation}
a^\top \nabla^2 \tilde g(u)\,a
=\frac{1}{\nu^4}\,\mathbb{E}\Bigl[\ell\bigl(f(u+V),x^\star\bigr)\,\Bigl(\bigl(a^\top V\bigr)^2-\nu^2\Bigr)\Bigr].
\label{eq:dir_second_deriv}
\end{equation}
By~\eqref{eq:bounded_loss}, $0\le \ell\le C$, hence
\begin{align}
\bigl|a^\top \nabla^2 \tilde g(u)\,a\bigr|
&\le \frac{C}{\nu^4}\,\mathbb{E}\Bigl[\Bigl|\bigl(a^\top V\bigr)^2-\nu^2\Bigr|\Bigr].
\label{eq:dir_bound_step1}
\end{align}
Since $a^\top V\sim \mathcal{N}\bigl(0,\nu^2\bigr)$, we can write $a^\top V=\nu Z$ with $Z\sim\mathcal{N}(0,1)$. Substituting into~\eqref{eq:dir_bound_step1} yields
\begin{align}
\bigl|a^\top \nabla^2 \tilde g(u)\,a\bigr|
&\le \frac{C}{\nu^4}\,\nu^2\,\mathbb{E}\bigl[|Z^2-1|\bigr]
= \frac{C}{\nu^2}\,\mathbb{E}\bigl[|Z^2-1|\bigr].
\label{eq:dir_bound_final}
\end{align}
Using the identity for symmetric matrices,
\[
\bigl\|\nabla^2 \tilde g(u)\bigr\|_{\mathrm{op}}
=
\sup_{\|a\|_2=1}\bigl|a^\top \nabla^2 \tilde g(u)\,a\bigr|,
\]
we obtain from~\eqref{eq:dir_bound_final} that
\[
\bigl\|\nabla^2 \tilde g(u)\bigr\|_{\mathrm{op}}
\le \frac{C}{\nu^2}\,\mathbb{E}\bigl[|Z^2-1|\bigr].
\]
A uniform operator-norm bound on the Hessian implies that $\nabla \tilde g$ is Lipschitz with constant $\beta$ given by the same bound, proving~\eqref{eq:beta_bound}. Finally, $\mathbb{E}\bigl[|Z^2-1|\bigr]<1$ (numerically $\approx 0.968$), which yields the strict inequality in~\eqref{eq:beta_bound}. \hfill $\blacksquare$

\subsection{Proof of Proposition~\ref{prop:pca_characterization}}
\label{app:proof_pca_char}

To prove Proposition~\ref{prop:pca_characterization}, consider and fix any unit vector $\delta$ with $\|\delta\|_2=1$. For every index $i$, define
\[
\psi_i(\alpha)\triangleq \alpha\,\delta^\top q_i - \frac{\beta}{2}\alpha^2.
\]
The inner maximization in~\eqref{eq:surrogate_opt} equals $\max_{\alpha\in\mathbb{R}}\psi_i(\alpha)$. Since $\psi_i$ is a concave quadratic, it has a unique maximizer given by the first-order necessary condition (FONC):
\[
\psi_i'(\alpha_i^*)=\delta^\top q_i - \beta \alpha_i^* = 0
\quad \Longrightarrow\quad
\alpha_i^\star=\frac{\delta^\top q_i}{\beta}.
\]
Substituting $\alpha_i^\star$,
\[
\max_{\alpha\in\mathbb{R}}\psi_i(\alpha)
=
\psi_i\bigl(\alpha_i^\star\bigr)
=
\frac{\bigl(\delta^\top q_i\bigr)^2}{2\beta}.
\]
Therefore, the objective in~\eqref{eq:surrogate_opt} is proportional to
\begin{align*}
\max_{\|\delta\|_2=1}\ \frac{1}{N}\sum_{i=1}^N \bigl(\delta^\top q_i\bigr)^2
&=
\max_{\|\delta\|_2=1}\ \delta^\top\Bigl(\frac{1}{N}\sum_{i=1}^N q_i q_i^\top\Bigr)\delta
=
\max_{\|\delta\|_2=1}\ \delta^\top \widehat\Sigma_q\,\delta.
\end{align*}
The maximizers of the Rayleigh quotient $\delta^\top \widehat\Sigma_q \delta$ over $\|\delta\|_2=1$ are precisely the top eigenvectors of $\widehat\Sigma_q$. \hfill $\blacksquare$

\subsection{Proof of Theorem~\ref{thm:uap_concentration}}
\label{app:proof_uap_conc}

To show the theorem, here we prove~\eqref{eq:conc_uap_obj}--\eqref{eq:conc_pca} and then apply a union bound.

\subsubsection{Proof of the objective concentration in~\eqref{eq:conc_uap_obj}}
Fix any perturbation $\delta$. By the bounded loss assumption $0\le \ell(\cdot)\le C$, we have for all $y$,
\[
0 \le g(y,\delta)=\mathbb{E}_{V}\bigl[\ell\bigl(f(y+V+\delta),x^\star\bigr)\bigr] \le C.
\]
Hence $Z_i\triangleq g\bigl(Y_i,\delta\bigr)$ are i.i.d.\ and lie in $[0,C]$ almost surely. Hoeffding's inequality implies
\[
\Pr\Bigl(\Bigl|\frac{1}{N}\sum_{i=1}^N Z_i - \mathbb{E}[Z_1]\Bigr|\ge t\Bigr)
\le 2\exp\Bigl(-\frac{2Nt^2}{C^2}\Bigr).
\]
Since $\frac{1}{N}\sum_{i=1}^N Z_i=\widehat F_N(\delta)$ and $\mathbb{E}[Z_1]=F(\delta)$,
\[
\Pr\bigl(\bigl|\widehat F_N(\delta)-F(\delta)\bigr|\ge t\bigr)
\le 2\exp\Bigl(-\frac{2Nt^2}{C^2}\Bigr).
\]
Setting the right-hand side to $\eta/2$ yields $t=C\sqrt{\log\bigl(4/\eta\bigr)/(2N)}$, proving~\eqref{eq:conc_uap_obj} with probability at least $1-\eta/2$.

\subsubsection{Matrix concentration: proof of~\eqref{eq:conc_sigma}}
Define $\Sigma_q=\mathbb{E}\bigl[q q^\top\bigr]$ and $\widehat\Sigma_q=\frac{1}{N}\sum_{i=1}^N q_i q_i^\top$, and let
\[
X_i \triangleq q_i q_i^\top - \Sigma_q.
\]
Then $\mathbb{E}[X_i]=0$ and $\widehat\Sigma_q-\Sigma_q=\frac{1}{N}\sum_{i=1}^N X_i$. By $\|q_i\|_2\le L$,
\[
\|q_i q_i^\top\|_{\mathrm{op}}=\|q_i\|_2^2 \le L^2,
\qquad
\|\Sigma_q\|_{\mathrm{op}}
=
\bigl\|\mathbb{E}[q q^\top]\bigr\|_{\mathrm{op}}
\le \mathbb{E}\bigl\|q q^\top\bigr\|_{\mathrm{op}}
= \mathbb{E}\|q\|_2^2
\le L^2,
\]
and thus $\|X_i\|_{\mathrm{op}}\le 2L^2$ almost surely.

Let $S\triangleq \sum_{i=1}^N X_i$. A standard matrix Hoeffding inequality for independent, mean-zero, self-adjoint matrices with $\|X_i\|_{\mathrm{op}}\le R$ gives
\[
\Pr\bigl(\|S\|_{\mathrm{op}}\ge t\bigr)
\le 2n\exp\Bigl(-\frac{t^2}{8\sum_{i=1}^N R^2}\Bigr)
=
2n\exp\Bigl(-\frac{t^2}{8NR^2}\Bigr).
\]
With $R=2L^2$ and $t=N\varepsilon$,
\[
\Pr\Bigl(\bigl\|\widehat\Sigma_q-\Sigma_q\bigr\|_{\mathrm{op}}\ge \varepsilon\Bigr)
=
\Pr\bigl(\|S\|_{\mathrm{op}}\ge N\varepsilon\bigr)
\le 2n\exp\Bigl(-\frac{N\varepsilon^2}{32L^4}\Bigr).
\]
Setting the right-hand side to $\eta/2$ yields
\[
\varepsilon = 4\sqrt{2}\,L^2\sqrt{\frac{\log\bigl(4n/\eta\bigr)}{N}},
\]
which proves~\eqref{eq:conc_sigma} with probability at least $1-\eta/2$.

\subsubsection{Principal component concentration: proof of~\eqref{eq:conc_pca}}
Let $E\triangleq \widehat\Sigma_q-\Sigma_q$ and $\Delta=\lambda_1(\Sigma_q)-\lambda_2(\Sigma_q)>0$. The Davis--Kahan $\sin\Theta$ theorem (for the top eigenvector) gives
\[
\sin\angle\bigl(\hat u_1,u_1\bigr)\le \frac{\|E\|_{\mathrm{op}}}{\Delta}
= \frac{\bigl\|\widehat\Sigma_q-\Sigma_q\bigr\|_{\mathrm{op}}}{\Delta}.
\]
Combining with~\eqref{eq:conc_sigma} yields~\eqref{eq:conc_pca}.

\subsubsection{Union bound}
The events for~\eqref{eq:conc_uap_obj} and~\eqref{eq:conc_sigma} each hold with probability at least $1-\eta/2$. By a union bound, they hold simultaneously with probability at least $1-\eta$, which implies~\eqref{eq:conc_pca} as well. \hfill $\blacksquare$

\section{Detailed Experimental Setup}
\label{app:exp_setup}

In this section, we provide a detailed description of the data generation process, the training details for neural decoders, the hyperparameter configurations for adversarial attacks.

\subsection{Data Generation and Code Construction}
\noindent\textbf{Channel Model.}
We simulate the transmission over an Additive White Gaussian Noise (AWGN) channel using Binary Phase Shift Keying (BPSK) modulation. Let $\mathbf{c} \in \{0, 1\}^n$ denote the codeword. The modulated signal is given by $\mathbf{x} = 1 - 2\mathbf{c}$, where $\mathbf{x} \in \{+1, -1\}^n$. The received signal $\mathbf{y}$ is defined as:
\begin{equation}
    \mathbf{y} = \mathbf{x} + \mathbf{z}, \quad \mathbf{z} \sim \mathcal{N}(0, \sigma^2 \mathbf{I}_n),
\end{equation}
where the noise variance $\sigma^2$ is determined by the Signal-to-Noise Ratio (SNR) via $\sigma^2 = (2 R \cdot 10^{\text{SNR}/10})^{-1}$, with code rate $R = k/n$.

\subsection{Training Details for Neural Decoders}
To ensure a fair comparison, the neural decoders (ECCT~\cite{choukroun2022ecct} and CrossMPT~\cite{parkcrossmpt}) are trained from scratch using the same configuration. For each code configuration $(n, k)$, we generate the training set following the settings in ECCT and CrossMPT. The training samples are generated with SNRs uniformly sampled from the range $[2\,\text{dB}, 8\,\text{dB}]$.

\section{Adversarial Attack Settings}

In this study, we evaluate the robustness of decoders against five distinct attack strategies. All attacks are constrained under the $L_2$-norm and the given energy constraint $\alpha$, where the perturbation $\delta$ must satisfy $||\delta||_2 \leq \alpha||y||_2$.

\noindent\textbf{Random Noise Baseline.} To distinguish adversarial attack methods from general sensitivity to channel noise, we choose a random noise baseline. This perturbation is generated by sampling a vector from a standard normal distribution $\mathcal{N}(0, I)$ and normalizing it to a given energy constraint $||\delta||_2 \leq \alpha||y||_2$.

\noindent\textbf{Fast Gradient Method (FGM).} FGM~\cite{goodfellow2014explaining} serves as a single-step attack that linearizes the loss function to generate adversarial examples. We implement the $L_2$-normalized version, where the perturbation is computed by taking the gradient of the loss with respect to the input, normalizing it to a unit vector, and scaling it to the energy constraint.

\noindent\textbf{Projected Gradient Descent (PGD).} PGD~\cite{madry2017towards} seeks to find the optimal perturbation within the $\epsilon$-ball. We set the attack with $T=20$ iterations. We initialize with a random perturbation and use a step size of $\alpha = 1.2 \times (\epsilon / 20)$, projecting the adversarial perturbation back onto the $\epsilon$-ball satisfying energy constraint.

\noindent\textbf{Universal Adversarial Perturbation (UAP).} UAP~\cite{moosavi2017universal} aims to find an input-agnostic perturbation vector that degrades performance across the entire dataset. We optimize this universal vector using Stochastic Gradient Descent Accumulation (SGDA) over a training set of 50 batches. The optimization uses a learning rate of $\eta=0.05$, with the perturbation clipped at each step to ensure it satisfies the enegy constraint.

\noindent\textbf{PCA-based UAP (UAP-PCA).} UAP-PCA~\cite{choi2022universal} applies Principal Component Analysis (PCA) to select the universal perturbation. We collect gradients from 50 batches of samples and apply PCA. The final attack vector is constructed by aligning the perturbation with the Top-1 principal component of the gradient matrix, scaled to the energy constraint.

\section{Additional Experimental Results}
\subsection{Ablation Study: Change of Energy Constraint}\label{appendix:ablation}
In this section, we present additional results on the impact of varying perturbation energy constraints on LDPC(121,60) codes, as shown in Figure~\ref{fig:ablation}. We observe a consistent trend: as the energy constraint increases, the performance degradation becomes more severe. Furthermore, adversarial attacks consistently induce significantly higher error rates compared to the random noise baseline.
\begin{figure}[htbp]
    \centering
    \includegraphics[width=0.95\linewidth]{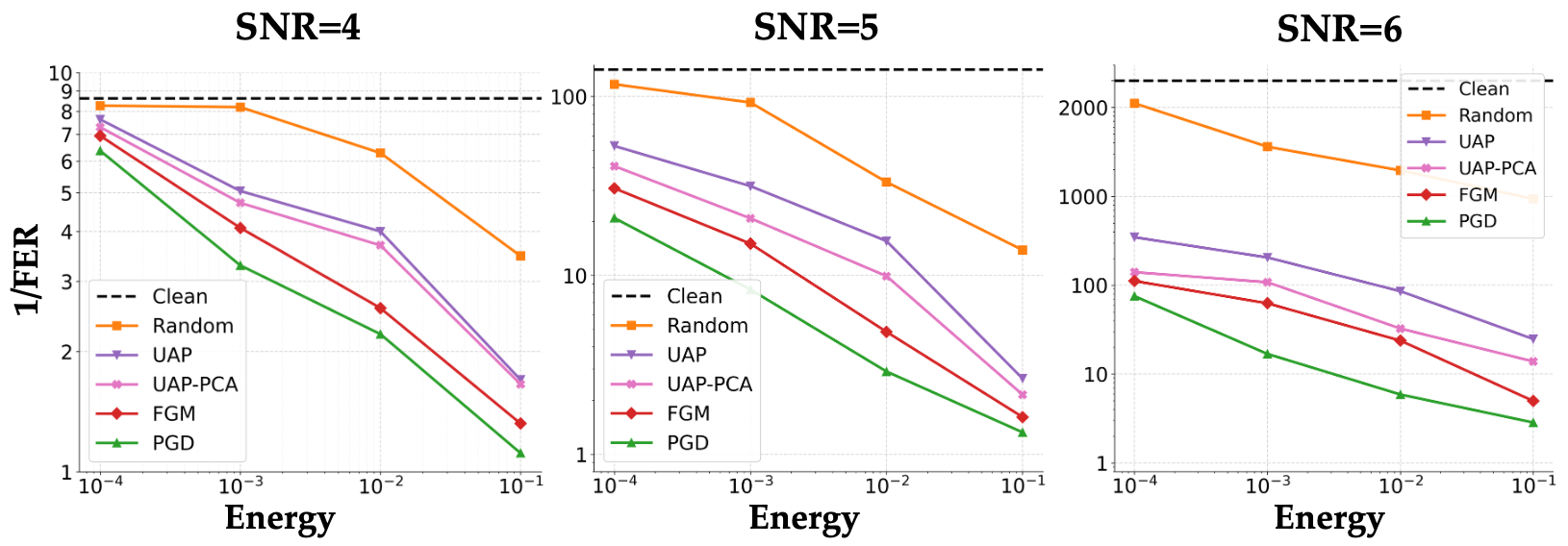}
    \caption{Comparison of \textbf{FER-Inverse (1/FER)} under various energy constraints for ECCT on LDPC(121,60).}
    \label{fig:ablation}
\end{figure}
\subsection{Transferability of Adversarial Perturbation}\label{appendix:transfer}
In this section, we present additional results regarding the transferability of adversarial perturbations. Specifically, we generate optimal perturbations using both the universal method (UAP) and the sample-level method (PGD) targeting the ECCT model on LDPC(121,60) (Table~\ref{tab:transfer}) and Polar(128,64) (Table~\ref{tab:transfer_polar}) codes. We then evaluate these perturbations against various baselines, including AI-based decoders (ECCT, CrossMPT, and the diffusion-based DDECC~\cite{choukroundenoising}) and traditional decoders (Min-sum, Sum-product for LDPC and SC (Successive Cancellation)~\cite{arikan2009channel}, Sum-product for Polar). The results demonstrate that adversarial attacks exhibit strong transferability across AI decoders, whereas traditional decoders remain robust, highlighting the inherent fragility of neural-based decoding architectures.
\begin{table}[H]
    \centering
    \caption{Evaluation of attack transferability measured by \textbf{FER-Inverse (1/FER)} across decoders on LDPC(121,60).}
    \renewcommand{\arraystretch}{1.125}
    \setlength{\tabcolsep}{4pt} 
    
    \resizebox{0.66\textwidth}{!}{%
        \begin{tabular}{c c l c c c}
            \toprule
            \textbf{Type} & \textbf{Decoder} & \textbf{Attack} & \textbf{4 dB} & \textbf{5 dB} & \textbf{6 dB} \\
            \midrule
            
            \multirow{12}{*}{\rotatebox{90}{\textbf{AI Decoder}}} 
            
            & \multirow{4}{*}{ECCT}     
              & Clean & $8.6 \times 10^{0}$ & $1.4 \times 10^{2}$ & $1.1 \times 10^{4}$ \\
            & & Random & $8.2 \times 10^{0}$ & $9.3 \times 10^{1}$ & $3.6 \times 10^{3}$ \\
            & & \textbf{UAP}   & $5.1 \times 10^{0}$ & $3.2 \times 10^{1}$ & $8.0 \times 10^{2}$ \\
            & & \textbf{PGD}   & $3.3 \times 10^{0}$ & $8.3 \times 10^{0}$ & $1.7 \times 10^{1}$ \\ 
            \cmidrule(l){2-6}
            
            & \multirow{4}{*}{CrossMPT} 
              & Clean & $2.1 \times 10^{1}$ & $6.1 \times 10^{2}$ & $1.7 \times 10^{4}$ \\
            & & Random & $1.9 \times 10^{1}$ & $2.7 \times 10^{2}$ & $3.8 \times 10^{3}$ \\
            & & \textbf{UAP}   & $1.7 \times 10^{1}$ & $6.7 \times 10^{1}$ & $8.8 \times 10^{2}$ \\
            & & \textbf{PGD}   & $9.2 \times 10^{0}$ & $3.2 \times 10^{1}$ & $5.5 \times 10^{1}$ \\
            \cmidrule(l){2-6}

            & \multirow{4}{*}{DDECC} 
              & Clean & $1.7\times10^{1}$ & $5.5 \times 10^{2}$ & $1.4 \times 10^{4}$ \\
            & & Random & $1.6\times10^{1}$ & $2.4 \times 10^{2}$ & $3.2 \times 10^{3}$ \\
            & & \textbf{UAP}   & $1.4\times10^{1}$ & $6.5 \times 10^{1}$ & $9.5 \times 10^{2}$ \\
            & & \textbf{PGD}   & $9.5\times10^{0}$ & $4.3 \times 10^{1}$ & $7.3 \times 10^{1}$ \\

            \midrule

            \multirow{8}{*}{\rotatebox{90}{\textbf{Traditional Decoder}}} 
            
            & \multirow{4}{*}{Min-sum}       
              & Clean & $5.4 \times 10^{0}$ & $5.1 \times 10^{1}$ & $2.0 \times 10^{3}$ \\
            & & Random & $5.4 \times 10^{0}$ & $5.1 \times 10^{1}$ & $1.9 \times 10^{3}$ \\ 
            & & \textbf{UAP}   & $5.3 \times 10^{0}$ & $5.0 \times 10^{1}$ & $1.7 \times 10^{3}$ \\
            & & \textbf{PGD}   & $5.0 \times 10^{0}$ & $4.3 \times 10^{1}$ & $1.3 \times 10^{3}$ \\
            \cmidrule(l){2-6}
            
            & \multirow{4}{*}{Sum-product}  
              & Clean & $1.6 \times 10^{1}$ & $4.9 \times 10^{2}$ & $1.5 \times 10^{4}$ \\
            & & Random & $1.6 \times 10^{1}$ & $4.8 \times 10^{2}$ & $1.4 \times 10^{4}$ \\
            & & \textbf{UAP}   & $1.5 \times 10^{1}$ & $4.1 \times 10^{2}$ & $1.1 \times 10^{4}$ \\
            & & \textbf{PGD}   & $1.3 \times 10^{1}$ & $3.3 \times 10^{2}$ & $1.0 \times 10^{4}$ \\
            
            \bottomrule
        \end{tabular}
    }
    \label{tab:transfer}
\end{table}

\begin{table}[H]
    \centering
    \caption{Evaluation of attack transferability measured by \textbf{FER-Inverse (1/FER)} across decoders on Polar(128,64).}
    \renewcommand{\arraystretch}{1.125}
    \setlength{\tabcolsep}{4pt} 
    
    \resizebox{0.66\textwidth}{!}{%
        \begin{tabular}{c c l c c c}
            \toprule
            \textbf{Type} & \textbf{Decoder} & \textbf{Attack} & \textbf{4 dB} & \textbf{5 dB} & \textbf{6 dB} \\
            \midrule
            
            \multirow{12}{*}{\rotatebox{90}{\textbf{AI Decoder}}} 
            
            & \multirow{4}{*}{ECCT}     
              & Clean & $8.5 \times 10^{0}$ & $7.4 \times 10^{1}$ & $1.4 \times 10^{3}$ \\
            & & Random & $8.1 \times 10^{0}$ & $5.9 \times 10^{1}$ & $6.1 \times 10^{2}$ \\
            & & \textbf{UAP}   & $5.4 \times 10^{0}$ & $3.4\times 10^{1}$ & $9.8 \times 10^{1}$ \\
            & & \textbf{PGD}   & $3.0 \times 10^{0}$ & $5.4 \times 10^{0}$ & $8.3 \times 10^{0}$ \\ 
            \cmidrule(l){2-6}
            
            & \multirow{4}{*}{CrossMPT} 
              & Clean & $2.5 \times 10^{1}$ & $2.4 \times 10^{2}$ & $1.0 \times 10^{4}$ \\
            & & Random & $2.2 \times 10^{1}$ & $1.8 \times 10^{2}$ & $5.8 \times 10^{3}$ \\
            & & \textbf{UAP}   & $1.9 \times 10^{1}$ & $8.4 \times 10^{1}$ & $9.2 \times 10^{2}$ \\
            & & \textbf{PGD}   & $5.1 \times 10^{0}$ & $9.1 \times 10^{0}$ & $4.9 \times 10^{1}$ \\
            \cmidrule(l){2-6}

            & \multirow{4}{*}{DDECC} 
              & Clean & $5.2 \times 10^{1}$ & $7.4\times10^{3}$ & $9.7\times10^{5}$ \\
            & & Random & $4.9 \times 10^{1}$ & $6.5 \times 10^{3}$ & $7.1 \times 10^{5}$ \\
            & & \textbf{UAP}   & $2.7\times 10^{1}$ & $2.1\times 10^{3}$ & $3.5\times 10^{4}$ \\
            & & \textbf{PGD}   & $9.7\times 10^{0}$ & $7.5\times 10^{2}$ & $2.2\times 10^{2}$ \\

            \midrule

            \multirow{8}{*}{\rotatebox{90}{\textbf{Traditional Decoder}}} 
            
            & \multirow{4}{*}{SC}       
              & Clean & $4.2 \times 10^{0}$ & $3.8 \times 10^{1}$ & $4.5 \times 10^{2}$ \\
            & & Random & $4.1 \times 10^{0}$ & $3.7 \times 10^{1}$ & $4.2 \times 10^{2}$ \\ 
            & & \textbf{UAP}   & $3.8 \times 10^{0}$ & $3.4 \times 10^{1}$ & $3.7 \times 10^{2}$ \\
            & & \textbf{PGD}   & $3.5 \times 10^{0}$ & $3.2 \times 10^{1}$ & $3.5 \times 10^{2}$ \\
            \cmidrule(l){2-6}
            
            & \multirow{4}{*}{Sum-product}  
              & Clean & $5.6 \times 10^{0}$ & $5.9 \times 10^{1}$ & $6.8 \times 10^{2}$ \\
            & & Random & $5.4 \times 10^{0}$ & $5.7 \times 10^{1}$ & $6.5 \times 10^{2}$ \\
            & & \textbf{UAP}   & $5.1 \times 10^{0}$ & $5.2 \times 10^{1}$ & $6.0 \times 10^{2}$ \\
            & & \textbf{PGD}   & $4.9 \times 10^{0}$ & $5.0 \times 10^{1}$ & $5.6 \times 10^{2}$ \\

            \bottomrule
        \end{tabular}
    }
    \label{tab:transfer_polar}
\end{table}

\end{document}